\documentclass[conference]{IEEEtran}
\IEEEoverridecommandlockouts
\usepackage[english]{babel}
\usepackage{amsmath,amsthm,mathrsfs}
\usepackage{amsfonts}
\usepackage[final,hyperfootnotes=false,hidelinks]{hyperref}
\usepackage[final]{graphicx}
\usepackage{epstopdf}
\usepackage{tikz}
\usepackage{pgfplots}
\usepackage{subfig}
\usepackage{flushend}
\usepackage{cprotect}
\usepackage{comment}

\usepackage{tikz}
\usetikzlibrary{arrows,positioning,fit,backgrounds}
\usetikzlibrary{calc}

\usetikzlibrary{arrows,positioning,fit,backgrounds,shapes}
\usetikzlibrary{calc}
\newtheorem{thm}{Theorem}
\newtheorem{cor}[thm]{Corollary}
\newtheorem{lem}[thm]{Lemma}
\newtheorem{prop}[thm]{Proposition}
\theoremstyle{definition}
\newtheorem{defn}[thm]{Definition}
\theoremstyle{remark}
\newtheorem{rem}[thm]{Remark}

\newcommand{\mb}{\mathbf}
\newcommand{\bsigma}{{\bf \Sigma}}
\newcommand{\half}{\frac{1}{2}}

\DeclareMathOperator*{\rmd}{d}
\DeclareMathOperator*{\rmds}{d^{\star}}

\begin{document}

\title{Smoothing Brascamp-Lieb Inequalities and Strong Converses for Common Randomness Generation\thanks{This work was supported in part by  NSF Grants CCF-1528132,
CCF-0939370 (Center for Science of Information),
CCF-1116013,
CCF-1319299,
CCF-1319304,
CCF-1350595 and AFOSR FA9550-15-1-0180.}%
}

\author{
\IEEEauthorblockN{Jingbo Liu\IEEEauthorrefmark{1}, Thomas~A.~Courtade\IEEEauthorrefmark{2}, Paul Cuff\IEEEauthorrefmark{1} and Sergio Verd\'{u}\IEEEauthorrefmark{1}}
\IEEEauthorblockA{ \IEEEauthorrefmark{1}Department of Electrical Engineering, Princeton University\\
 \IEEEauthorrefmark{2}Department of Electrical Engineering and Computer Sciences, University of California, Berkeley \\
 Email:  \{{jingbo,cuff,verdu}\}@princeton.edu, courtade@eecs.berkeley.edu
           \vspace{-2ex}   }
}
\maketitle

\begin{abstract}
We study the infimum of the best constant in a functional inequality, the Brascamp-Lieb-like inequality, over auxiliary measures within a neighborhood of a product distribution.
In the finite alphabet and the Gaussian cases, such an infimum converges to the best constant in a mutual information inequality.
Implications for strong converse properties of two common randomness (CR) generation problems are discussed.
In particular, we prove the strong converse property of the rate region for the omniscient helper CR generation problem in the discrete and the Gaussian cases.
The latter case is perhaps the first instance of a strong converse for a continuous source when the rate region involves auxiliary random variables.
\end{abstract}

\section{Introduction}
In the last few years, information theory has seen vibrant developments
in the study of the non-vanishing error probability regime,
and in particular, the successes in applying normal approximations
to gauge the back-off from the asymptotic limits as a function of delay.
Extending the achievements for point-to-point communication systems in \cite{hayashi2009information}\cite{polyanskiy2010channel}\cite{kostina2012fixed}
to network information theory problems
usually requires new ideas for proving tight non-asymptotic bounds. For achievability, single-shot covering lemmas and packing lemmas \cite{verdu2012non}\cite{liu_marton} supply convenient tools for distilling single-shot achievability bounds from the classical asymptotic achievability proofs.
These single-shot bounds are easy to analyze in the stationary memoryless case by choosing the auxiliary random variables to be i.i.d.~and applying the law of large numbers or the central limit theorem.

In contrast, there are few examples of single-shot converse bounds in the network setting.  Indeed, unlike their achievability counterparts, single-shot converses are often non-trivial to single-letterize to a strong converse.
In fact, there are few methods for obtaining strong converses for network information theory problems whose single-letter solutions involve auxiliaries; see e.g.~\cite[Section~9.2 ``Open problems and challenges ahead'']{CIT-086}.
Exceptions include the strong converses for select source networks \cite{csiszar2011information} where the method of types plays a pivotal role.

In this paper, through the example of a common randomness (CR) generation problem \cite[Theorem~4.2]{ahlswede1998common}, we demonstrate the power
of
a functional inequality, the \emph{Generalized Brascamp-Lieb-like (GBLL) inequality} \cite{lccv2015}:
\begin{align}
\int\exp\left(\sum_{j=1}^m\mathbb{E}[\log f_j(Y_j)|X=\cdot]-d\right){\rm d}\mu
&\le \prod_{j=1}^m\|f_j\|_{\frac{1}{c_j}},
\label{e_func}
\end{align}
in proving single-shot converses for problems involving multiple sources.
Here $\mu$, $(Q_{Y_j|X})$, $(\nu_j)$, $(c_j)$, $d$ are given and
$\|f_j\|_{\frac{1}{c_j}}:=\left(\int f_j^{1/c_j}{\rm d}\nu_j\right)^{c_j}$.
The key tool for single-letterizing such single-shot converses to strong converses is the ``achievability'' of the following problem: infimize the best constant $d$ in \eqref{e_func}
with the substitutions $\mu\leftarrow \mu_n$, $\nu_j\leftarrow\nu_j^{\otimes n}$ and $Q_{Y_j|X}\leftarrow Q_{Y_j|X}^{\otimes n}$,
where
the auxiliary measure $\mu_n$ is within a neighborhood (say in total variation) of $\mu^{\otimes n}$.
Interestingly, a product $\mu_n$ is generally not a good choice.
On the surface, this is reminiscent of the smooth R\'{e}nyi entropy \cite{renner2005simple}, who showed that the infimum (resp.~supremum) of the R\'{e}nyi entropy of order $\alpha<1$ (resp.~$\alpha>1$) of an auxiliary measure with a neighborhood of a product distribution behaves like the Shannon entropy. In reality, the smooth version of GBLL appears to be a much deeper problem, since structure at a finer resolution than weak typicality is involved.

The general philosophy appears to be that under certain regularity conditions,
$\frac{d}{n}$ (where $d$ is the best constant in the setting of product measures and smoothing above) converges to the best constant in a mutual information inequality.
We provide a general approach for verifying this principle, and apply it to the discrete memoryless and the Gaussian source.
When this principle holds, our single-shot converse proves the strong converse for the CR generation problem.

The proposed approach to strong converses has two main advantages compared with the method of types approach in \cite{csiszar2011information},
which are nicely illustrated by the example of CR generation:
1) The argument covers possibly stochastic decoders.
2) As illustrated by the Gaussian example, the approach is applicable to some non-discrete sources where the method of types is futile.
This is perhaps the first instance of a strong converse for a continuous source when the rate region involves auxiliaries.
We also refine the analysis to bound the second order rate.

In addition, we discuss the ``converse'' part of smooth BLL, which
generally follows from the achievability of CR generation problems. In fact, smooth BLL and CR generation may be considered as dual problems where the achievability of one implies the converse of the other, and vice versa.\footnote{Another example of such ``dual problems'' is channel resolvability and identification coding \cite{han1993approximation}.} 

It is also interesting to note that for hypercontractivity, which is a special case of the BLL inequality with the best constant being zero, Anantharam et~al.~\cite{anan_13} showed the equivalence between a relative entropy inequality and a mutual information inequality. This equivalence is lost for positive best constants. Thus smooth BLL is a conceptually satisfying way to regain the connection between these two inequalities.

Omitted proofs are given in the appendices of \cite{lccv_smooth2016}.

\section{Preliminaries}
\begin{defn}\label{defn1}
Given a nonnegative $\mu$ on $\mathcal{X}$, $\nu_j$ on $\mathcal{Y}_j$, and random transformations $Q_{Y_j|X}$, and $c_j\in(0,\infty)$, $j\in\{1,\dots,m\}$, define
\begin{align}
\rmd(\mu,(Q_{Y_j|X}),(\nu_j),c^m)
:=\sup\left\{\sum_{l=1}^m c_l D(P_{Y_l}\|\nu_j)-D(P_X\|\mu)\right\}
\nonumber
\end{align}
where the sup is over $P_X\ll \mu$ and $P_X\to Q_{Y_j|X}\to P_{Y_j}$.
\end{defn}
We shall abbreviate the notation in Definition~\ref{defn1} as $\rmd(\mu,\nu_j,c^m)$ when there is no confusion.

Note that $\mu$ and $\nu_j$ are not necessarily probability measures, and $\mu\to Q_{Y_j|X}\to \nu_j$ need not hold. These liberties are useful, e.g. in the proof of Theorem~\ref{thm_gauss}.
Generalizing an approach in \cite{carlen2009subadditivity}, we established the following \cite{lccv2015}:
\begin{prop}\label{prop_func}
Under the assumptions of Definition~\ref{defn1}, $\rmd(\cdot)$ is the minimum $d$ such that
\eqref{e_func} holds
for all nonnegative measurable functions $f_j$.
\end{prop}
We call \eqref{e_func} a \emph{generalized Brascamp-Lieb-like inequality} (GBLL).
The case of deterministic $Q_{Y_j|X}$  was considered in \cite{carlen2009subadditivity}, which we shall call a \emph{Brascamp-Lieb-like inequality} (BLL).
In the special case where $Q_{Y_j|X}$'s are a linear projections and $\mu$ and $\nu_j$ are Gaussian or Lebesgue, \eqref{e_func} is called a Brascamp-Lieb inequality; it is well-known that a Brascamp-Lieb inequality holds for a specific value of $d$ if and only if it holds for all Gaussian functions $(f_j)$ \cite{brascamp1976best}.

\begin{defn}
For nonnegative measures $\nu$ and $\mu$ on the same measurable space $(\mathcal{X},\mathscr{F})$ and $\gamma\in[1,\infty)$,
the $E_{\gamma}$ divergence is defined as
\begin{align}
E_{\gamma}(\nu\|\mu):=\sup_{\mathcal{A}\in\mathscr{F}}
\{\nu(\mathcal{A})-\gamma \mu(\mathcal{A})\}.
\end{align}
\end{defn}
Note that under this definition $E_1(P\|\mu)$ does not equal $\frac{1}{2}|P-\mu|$ if $\mu$ is not a probability measure. Properties of $E_{\gamma}$ used in this paper can be found in \cite{liu2015_egamma_arxiv}.

\begin{defn}
For $\delta\in[0,1)$, $Q_X$, $(Q_{Y_j|X})$ and $(\nu_j)$, define
\begin{align}
{\rm d}_{\delta}(Q_X,\nu_j,c^m):= \inf_{\mu\colon E_1(Q_X\|\mu)\le \delta}{\rm d}(\mu,\nu_j,c^m).
\label{e_smoothconst}
\end{align}
In the stationary memoryless case, define the \emph{$\delta$-smooth GBLL rate}\footnote{As is clear from the context, the random transformations implicit on the right side of \eqref{e6} are $(Q_{Y_j|X}^{\otimes n})$.}
\begin{align}
{\rm D}_{\delta}(Q_X,\nu_j,c^m):=
\limsup_{n\to\infty}\frac{1}{n}{\rm d}_{\delta}(Q_X^{\otimes n},
\nu_j^{\otimes n},c^m),
\label{e6}
\end{align}
and the \emph{smooth GBLL rate} is the limit
\begin{align}
{\rm D}_{0^+}(Q_X,\nu_j,c^m):=\lim_{\delta\downarrow0}{\rm D}_{\delta}(Q_X,\nu_j,c^m).
\end{align}
\end{defn}
\begin{rem}
Allowing unnormalized measures avoids the unnecessary step of normalization in the proof, and is in accordance with the literature on smooth R\'{e}nyi entropy, where such a relaxation generally gives rise to nicer properties and tighter non-asymptotic bounds, cf.~\cite{renner2005simple}\cite{liu2015_egamma_arxiv}.
\end{rem}

\begin{defn}\label{defn_dstar}
Given $Q_X$, $(Q_{Y_j|X})$ and $c^m\in(0,\infty)^m$, define
\begin{align}
\rmds(Q_X,c^m):=\sup_{P_{U|X}}\left\{\sum_{l=1}^m c_l I(U;Y_l)-I(U;X)\right\}.
\end{align}
We say $Q_X$, $(Q_{Y_j|X})$ and $(c_j)$ satisfy the \emph{$\delta$-smooth property} if
\begin{align}
{\rm D}_{\delta}(Q_X,Q_{Y_j},c^m)
=\rmds(Q_X,c^m),
\label{e_smooth}
\end{align}
\emph{(weak) smooth property} if ${\rm D}_{0^+}(Q_X,Q_{Y_j},c^m)
=\rmds(Q_X,c^m)$,
and \emph{strong smooth property} if \eqref{e_smooth} holds for all $\delta\in(0,1)$.
\end{defn}
From these definitions and a tensorization property of $\rmd(\cdot)$ \cite{lccv2015} we clearly have
\begin{align}
\rmd(Q_X,Q_{Y_j},c^m)
={\rm D}_0(Q_X,Q_{Y_j},c^m)
&\ge {\rm D}_{\delta}(Q_X,Q_{Y_j},c^m).
\end{align}
The goal is to explore conditions for ${\rm D}_{\delta}(Q_X,Q_{Y_j},c^m)=\rmds(Q_X,c^m)$.

\section{Achievabilities for Smooth GBLL}
Under various conditions, we provide upper bounds on ${\rm D}_{\delta}(Q_X,Q_{Y_j},c^m)$, establishing the achievability part of the strong smooth property.
\subsection{Hypercontractivity}
\label{sec_hypercontractivity}
If $\rmds(Q_X,c^m)=0$,
by an extension of the proof of equivalent formulations of hypercontractivity \cite{anan_13} we also have
$
\rmd(Q_X,Q_{Y_j},c^m)=0
$,
establishing that ${\rm D}_{0}(Q_X,Q_{Y_j},c^m)=\rmds(Q_X,c^m)$.

\subsection{Finite $|\mathcal{X}|$, and Beyond}
The main objective of this section is to show that
\begin{thm}\label{thm_discrete}
${\rm D}_{0^+}(Q_X,Q_{Y_j},c^m)\le\rmds(Q_X,c^m)$ if $\mathcal{X}$ is finite.
\end{thm}
We present a general method of proving achievability of smooth GBLL which, although not intuitive at the first sight, turns out to be successful for the distinct cases of the discrete and Gaussian sources. The following tensorization result is useful:
\begin{lem}\label{lem_singleletter}
Suppose $\tau_{\alpha}\colon \mathcal{X}\to\mathbb{R}$ is measurable for each (abstract) index $\alpha\in\mathcal{A}$. Fix any $\epsilon\in(0,1)$, and for each $n\in\{1,\dots\}$ define $g(n)$ as the supremum of
\begin{align}
\frac{1}{n} \left[
\sum_j c_j D(P_{Y^n|U}\|\nu_j^{\otimes n}|P_U)-D(P_{X^n|U}\|\mu^{\otimes n}|P_U)\right]
\label{e_f}
\end{align}
over $P_{UX^n}$ such that
$
\mathbb{E}\left[\frac{1}{n}\sum_{i=1}^n \tau_{\alpha}(\hat{X}_i)\right]\le\epsilon
$,
where
$\hat{X}^n\sim P_{X^n}$
and
$P_{UX^nY^n}:=P_{UX^n}Q_{Y|X}^{\otimes n}$. Then $g(n)\le g(1)$.
\end{lem}
The functions $\tau_{\alpha}(\cdot)$ can be thought of as (possibly negative) cost functions that enforce the $P_{UX}$
maximizing \eqref{e_f} to satisfy $P_X\approx Q_X$. If the probability that an i.i.d.~sequence induces a small cost is large, then one can choose the $\mu$ in the definition of the smooth property to be the restriction\footnote{In this paper, by restriction of a measure on a set we mean the result of cutting off the measure outside that set (without renormalizing).} of $Q_X^{\otimes n}$ on such a set. Therefore the following lemma will be the key to our proofs of the smooth property:
\begin{lem}\label{lem_key}
Suppose $\tau_{\alpha}$ is as in Lemma~\ref{lem_singleletter} and
define
\begin{align}
\mathcal{S}_{\epsilon}^n
:=\left\{x^n\colon\frac{1}{n}\sum_{i=1}^n \tau_{\alpha}(x_i)\le \epsilon\right\}.
\label{e_s}
\end{align}
If $P_{X^n}$ is supported on $\mathcal{S}_{\epsilon}^n$ for each $n$, then
\begin{align}
&\quad\limsup_{n\to\infty}\frac{1}{n}
\left[\sum_j c_jD(P_{Y_j^n}\|\nu_j^{\otimes n})
-D(P_{X^n}\|\mu^{\otimes n})\right]
\nonumber\\
&\le
\sup
\left\{\sum c_j D(P_{Y_j|U}\|\nu_j|P_U)
-D(P_{X|U}\|\mu|P_U)\right\}
\label{e_14}
\end{align}
where the sup on the right is over $P_{UX}$ such that $\mathbb{E}[\tau_{\alpha}(\hat{X})]\le\epsilon$.
\end{lem}
A remarkable aspect of Lemma~\ref{lem_key} is that the left side of \eqref{e_14}, which is a multi-letter quantity from the definition of $\rmd(\cdot)$, is upper bounded by a single-letter quantity.

\begin{lem}\label{lem_weakcont}
Suppose $(\mathcal{X},\mathscr{F})$ is a second countable topological space and $Q_X$ is a Borel measure. Define
\begin{align}
\sigma\colon P_X\mapsto\sum c_j D(P_{Y_j}\|Q_{Y_j})
-D(P_{X}\|Q_X).
\label{e_semi}
\end{align}
If $\phi$, the concave envelope of $\sigma$, is upper semicontinuous at $Q_X$,
then ${\rm D}_{0^+}(Q_X,Q_{Y_j},c^m)\le\rmds(Q_X,c^m)$.
\end{lem}
\begin{rem}
If $c_1=\dots =c_m=0$, then $\phi(P_X)=-D(P_{X}\|Q_X)$ always satisfies the upper semicontinuity in Lemma~\ref{lem_weakcont} because of the weak semicontinuity of the relative entropy. On the other hand, taking $m=1$, $c_1=2$, $Q_X$ any distribution on a countably infinite alphabet with $H(Q_X)<\infty$,
and $Q_{Y_1|X}$ the identity transformation,
we see $\sigma(P_X)=H(P_X)+D(P_X\|Q_X)$ and the upper semicontinuity condition in Lemma~\ref{lem_weakcont} fails.
\end{rem}

\begin{proof}[Proof of Theorem~\ref{thm_discrete}]
Assume w.l.o.g.~that $Q_X(x)>0,\,\forall x$ since otherwise we can delete $x$ from $\mathcal{X}$. Then $Q_X$ is in the interior of the probability simplex. Moreover $\phi(\cdot)$ in Lemma~\ref{lem_weakcont} is clearly bounded.
Thus by \cite[Corollary~7.4.1]{rockafellar2015convex},
the weak semicontinuity in Lemma~\ref{lem_weakcont} is fulfilled.
\end{proof}
\begin{rem}
For general $\mathcal{X}$, one cannot use the property of convex functions to conclude the semicontinuity as in the proof of Theorem~\ref{thm_discrete}. In fact, whenever $|\mathcal{X}|=\infty$, there are points in $\mathcal{X}$ with arbitrarily small probability, thus $Q_X$ cannot be in the interior of the probability simplex even under the stronger topology of total variation.
\end{rem}

\subsection{Gaussian Case}
The semicontinuity assumption in
Lemma~\ref{lem_weakcont}
appears too strong
for the case of the Gaussian distribution, which has a non-compact support. Nevertheless, we can proceed by picking a different $\tau_{\alpha}(\cdot)$ in Lemma~\ref{lem_key}.
\begin{thm}\label{thm_gauss}
${\rm D}_{0^+}(Q_X,Q_{Y_j},c^m)\le\rmds(Q_X,c^m)$ if $Q_{\bf X}$ and $(Q_{{\bf Y}_j|\bf X})$ are Gaussian.
\end{thm}
The proof hinges on our prior result \cite{lccv2015} about the Gaussian optimality in an optimization under a covariance constraint: suppose  $\mu$ and $\nu_j$ are the Lebesgue measures. Define
\begin{align}
F({\bf M})
&:=\sup\left\{-\sum c_j h({\bf Y}_j|U)+h({\bf X}|U)\right\}
\label{e14}
\\
=&\sup
\left\{\sum c_j D(P_{{\bf Y}_j|U}\|\nu_j|P_U)
-D(P_{\mb{X}|U}\|\mu|P_U)\right\}
\label{e15}
\end{align}
where the supremums are over $P_{U\bf X}$ such that $\bsigma_{\bf X}\preceq {\bf M}$.
Also suppose w.l.o.g.~that ${\bf X}\sim \mathcal{N}({\bf 0},\bsigma)$ under $Q_{\bf X}$.
\begin{prop}[\cite{lccv2015}]
\label{prop14}
$F(\mb{M})$ equals the sup in \eqref{e15} restricted to constant $U$ and Gaussian $\mb{X}$, which implies that
\begin{align}
F(\bsigma)+C=
\rmds(Q_{\bf X},Q_{{\bf Y}_j},c^m)
\label{e16}
\end{align}
where
\begin{align}
C:=\sum_j c_j h({\bf Y}_j)-h({\bf X}_j).
\label{e_c}
\end{align}
\end{prop}
\begin{proof}[Proof of Theorem~\ref{thm_gauss}]
Put $\mathcal{A}$ as the set of unit length vectors in $\mathcal{X}$ (a Euclidean space), and for each $\alpha\in\mathcal{A}$ define
$
\tau_{\alpha}(\mb{x}):=(\alpha^{\top}\bsigma^{-\half}{\bf x})^2-1
$.
Now, observe that for ${\bf x}^n\in\mathcal{X}^n$,
\begin{align}
\frac{1}{n}\sum_i\tau_{\alpha}(\mb{x}_i)
:=\alpha^{\top}\bsigma^{-\half}\left(\frac{1}{n}\sum_i{\bf x}
{\bf x}^{\top}\right)
\bsigma^{-\half}\alpha
-1,
\end{align}
so $\frac{1}{n}\sum_i\tau_{\alpha}(\mb{x}_i)\le\epsilon_1$ for all $\alpha\in\mathcal{A}$ is equivalent to the bound on the empirical covariance:
$
\frac{1}{n}\sum_i{\bf x}
{\bf x}^{\top}
\preceq (1+\epsilon_1)\bsigma
$.
Consider also the ``weakly typical set'' $\mathcal{T}_{\epsilon_2}^n$, defined as the set of sequences $\mb{x}^n$ such that
\begin{align}
&\frac{1}{n}\sum_i\left[\imath_{Q_{\bf X}\|\mu}(\mb{x}_i)-\sum_j c_j\mathbb{E}[\imath_{Q_{\mb{Y}_j}\|\nu_j}(\mb{Y}_j)|\mb{X}=\mb{x}_i]
\right]
\le
C+\epsilon_2
\end{align}
where $C$ was defined in \eqref{e_c}.
Now set $\mu_n$ as the restriction of $Q_{\bf X}^{\otimes n}$ on $\mathcal{S}_{\epsilon_1}^n\cap \mathcal{T}_{\epsilon_2}^n$. If $P_{{\bf X}^n}\ll \mu_n$, by Lemma~\ref{lem_key} we have
\begin{align}
\limsup_{n\to\infty}\frac{1}{n}
\left[\sum_j c_jD(P_{{\bf Y}_j^n}\|\nu_j^{\otimes n})
-D(P_{{\bf X}^n}\|\mu^{\otimes n})\right]
\le F((1+\epsilon_1)\bsigma).
\label{e_opt}
\end{align}
Since $P_{{\bf X}^n}$ is supported on $\mathcal{T}_{\epsilon_2}^n$, we also have
\begin{align}
&\frac{1}{n}
\left[\sum_j c_jD(P_{{\bf Y}_j^n}\|\nu_j^{\otimes n})
-D(P_{{\bf X}^n}\|\mu^{\otimes n})\right]\nonumber
+C
\nonumber
\\
&\ge\frac{1}{n}
\left[\sum_j c_jD(P_{{\bf Y}_j^n}\|Q_{\mb{Y}_j}^{\otimes n})
-D(P_{{\bf X}^n}\|Q_{\bf X}^{\otimes n})\right]-\epsilon_2
\label{e33}
\end{align}
Hence from \eqref{e_opt}-\eqref{e33} we conclude
\begin{align}
&\limsup_{n\to\infty}\frac{1}{n}
\left[\sum_j c_jD(P_{{\bf Y}_j^n}\|Q_{\mb{Y}_j}^{\otimes n})
-D(P_{{\bf X}^n}\|\mu_n)\right]
\nonumber
\\
&\le
F((1+\epsilon_1)\bsigma)+C+\epsilon_2
\label{e34}
\end{align}
where we used $D(P_{{\bf X}^n}\|Q_{\bf X}^{\otimes n})=D(P_{{\bf X}^n}\|\mu_n)$.
Also, by the law of large numbers, $\lim_{n\to\infty} Q_{\bf X}^{\otimes n}(\mathcal{S}_{\epsilon_1}^n\cap \mathcal{T}_{\epsilon_2}^n)=1$ so
$
\lim_{n\to\infty} E_1(Q_{\bf X}^{\otimes n}\|\mu_n)=1
$.
Thus \eqref{e34}, Proposition~\ref{prop14} and the continuity of $F$ (which can be verified since \eqref{e14} is essentially a matrix optimization problem) imply the desired result.
\end{proof}

\section{Converse for the One-Communicator Problem}
\label{sec_onecommunicator}
We prove a single-shot bound connecting smooth GBLL and  one-communicator CR generation~\cite[Theorem~4.2]{ahlswede1998common}, allowing us to prove the converse of one using the achievability of the other.

Let $Q_{XY^m}$ be the joint distribution of sources $X$, $Y_1$, \dots, $Y_m$, observed by terminals ${\sf T}_0$, \dots, ${\sf T}_m$ as shown in Figure~\ref{f_1com}.
The communicator ${\sf T}_0$ computes the integers $W_1(X)$, \dots, $W_m(X)$ and sends them to ${\sf T}_1$, \dots, ${\sf T}_m$, respectively. Then, terminals ${\sf T}_0$, \dots, ${\sf T}_m$ compute integers $K(X)$, $K_1(Y_1,W_1)$,\dots, $K_m(Y_m,W_m)$.
The goal is to produce $K=K_1=\dots=K_m$ with high probability
with $K$ almost equiprobable.

\begin{figure}[h!]
  \centering
\begin{tikzpicture}
[scale=2,
      dot/.style={draw,fill=black,circle,minimum size=0.7mm,inner sep=0pt},arw/.style={->,>=stealth}]
  \node[rectangle,draw,rounded corners] (A) {${\sf T}_1$};
  \node[rectangle,draw,rounded corners] (B) [right= 1.4cm of A] {${\sf T}_2$};
  \node[rectangle] (C) [right =of B] {$\dots$};
  \node[rectangle,draw,rounded corners] (D) [right =of C] {${\sf T}_m$};
  \node[rectangle,draw,rounded corners] (T) [above right=of B, xshift=-13mm, yshift=10mm] {${\sf T}_0$};
  \node[rectangle] (Z) [left=0.4cm of T] {$X$};
  \node[rectangle] (K1) [below =0.4cm of A] {$K_1$};
  \node[rectangle] (K2) [below =0.4cm of B] {$K_2$};
  \node[rectangle] (Km) [below =0.4cm of D] {$K_m$};
  \node[rectangle] (K) [above =0.4cm of T] {$K$};
  \node[rectangle] (X1) [left =0.4cm of A] {$Y_1$};
  \node[rectangle] (X2) [left =0.4cm of B] {$Y_2$};
  \node[rectangle] (Xm) [left =0.4cm of D] {$Y_m$};
 \draw[arw] (Z) to node[]{} (T);
 \draw[arw] (X1) to node[]{} (A);
  \draw[arw] (X2) to node[]{} (B);
   \draw[arw] (Xm) to node[]{} (D);
  \draw [arw] (A) to node[midway,above]{} (K1);
  \draw [arw] (B) to node[midway,above]{} (K2);
  \draw [arw] (D) to node[midway,above]{} (Km);
  \draw [arw] (T) to node[]{} (K);
  \draw [arw,line width=1.5pt] (T) to node[midway,left]{$W_1$} (A.north);
  \draw [arw,line width=1.5pt] (T) to node[midway,left]{$W_2$} (B.north);
  \draw [arw,line width=1.5pt] (T) to node[midway,left]{$W_m$} (D.north);
\end{tikzpicture}
\caption{CR generation with one-communicator}
\label{f_1com}
\end{figure}
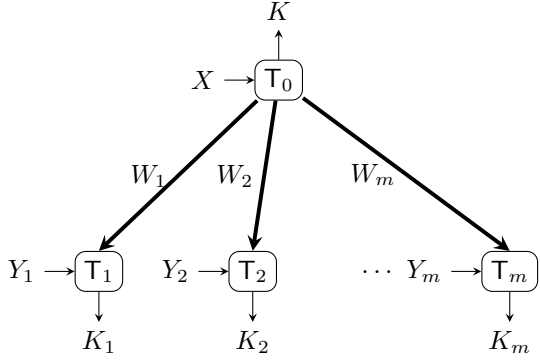
In the stationary memoryless case, put $X\leftarrow X^n$, $Y_j\leftarrow Y_j^n$. Denote by $R$ and $R_j$ the rates of $K$ and $W_j$, respectively.
Under various performance metrics (cf.~\cite{ahlswede1998common}\cite{liu2015key}), the achievable region is the set of $(R,R_1,\dots,R_m)$ such that
\begin{align}
\rmds(Q_X,c^m)+\sum_j c_jR_j\ge \left(\sum_jc_j-1\right)R
\label{e36}
\end{align}
for all $c^m\in(0,\infty)^m$. \footnote{Remark in passing that the corresponding \emph{key} generation problem, which places the additional constraint that $W_j\perp K$ asymptotically for each $j$, is solved in \cite{liu2015key} with a different rate region involving $m+1$ auxiliaries.}
\begin{thm}[Strong converse for one-communicator CR generation]
\label{thm_onecommunicator}
For finite $|\mathcal{X}|$, $|\mathcal{Y}_1|,\dots,|\mathcal{Y}_m|$,
suppose $(R,R_1,\dots,R_m)$ fails \eqref{e36} for some $c^m$.
If $(\delta_1,\delta_2)$ is such that
 \begin{align}
\mathbb{P}[K=K_1=\dots=K_m]
&\ge 1-\delta_1;
\label{e_38}
\\
\frac{1}{2}|Q_K-T_K|
&\le\delta_2
\label{e_39}
\end{align}
can hold for some CR generation scheme at rates $(R,R_1,\dots,R_m)$  for sufficiently large $n$
where $T_K$ is the equiprobable distribution on $\mathcal{K}$,
then
$
\delta_1+\delta_2\ge 1
$.
\end{thm}
The following lemma establishes a \emph{single-shot} connection between one-communicator CR generation and smooth GBLL, which allows us to prove the converse of one problem from the achievability of the other.
For simplicity of presentation, we state it in the case of $m=1$.\footnote{Note that this problem is unlike the usual ``image-size characterization'' \cite[Chapter~15]{csiszar2011information} which is difficult to generalize to $m\ge 3$ case.}
\begin{lem}\label{lem_connect}
Suppose that there exist $\delta_1,\delta_2\in(0,1)$,
a stochastic encoder $Q_{W|X}$, and deterministic decoders $Q_{K|X}$ and $Q_{\hat{K}|WY}$, such that \eqref{e_38} and \eqref{e_39} hold.
Also, suppose that there exist $\mu_X$, $\delta,\epsilon,\epsilon'\in(0,1)$ and $c,d\in(0,\infty)$ such that
\begin{align}
E_1(Q_X\|\mu_X)
&\le \delta;
\label{e_neighbor}
\\
\mu_X\left(x\colon Q_{Y|X=x}(\mathcal{A})\ge 1-\epsilon'\right)
&\le
2^c\exp(d) Q_Y^{c(1-\epsilon)}(\mathcal{A})
\label{e39}
\end{align}
for any $\mathcal{A}\subseteq \mathcal{Y}$.
Then, for any $\delta_3,\delta_4\in(0,1)$ such that $\delta_3\delta_4=\delta_1+\delta$, we have
\begin{align}
\delta_2\ge 1-\delta-\delta_3-\frac{1}{|\mathcal{K}|}
-\frac{2^{\frac{1}{1-\epsilon}}\exp\left(\frac{d}{c(1-\epsilon)}\right)
|\mathcal{W}|}{(\epsilon'-\delta_4)
^{\frac{1}{c(1-\epsilon)}}|\mathcal{K}|^{1-\frac{1}{c(1-\epsilon)}}}.
\label{e_41}
\end{align}
\end{lem}

\begin{rem}\label{rem15}
The relevance of the Lemma~\ref{lem_connect} to smooth GBLL is seen by setting
$
f(y):=(1_{\mathcal{A}}(y)+ Q_Y(\mathcal{A}) 1_{\mathcal{\bar{A}}}(y))^c
$
in \eqref{e_func}. We then see \eqref{e39} holds for any $\epsilon=\epsilon'\in(0,1)$.
\end{rem}
\begin{rem}\label{rem_16}
In the stationary memoryless case $Q_X\leftarrow Q_X^{\otimes n}$, $Q_{Y|X}\leftarrow Q_{Y|X}^{\otimes n}$, suppose $|\mathcal{X}|,|\mathcal{Y}|<\infty$.
Using the blowing-up lemma \cite{ahlswede1976bounds}, we can show that
for any $\delta,\epsilon,\epsilon'\in(0,1)$ and $d>\rmds(Q_X,c)$, there exists $n$ large enough such that \eqref{e39} is satisfied with $d\leftarrow nd$ for some $\mu_X$ (more precisely, the restriction of $Q_X^{\otimes n}$ on a strongly typical set) satisfying \eqref{e_neighbor}.
\end{rem}
\begin{proof}[Proof of Theorem~\ref{thm_onecommunicator}]
Again consider $m=1$ case for simplicity.
Suppose that $(R,R_1)$ is such that
\eqref{e36} fails for some $c>0$. Then, there is $\epsilon\in(0,1)$ and $d>\rmds(Q_X,c)$ such that
\eqref{e_54} does not hold.
If we choose $\delta>0$ arbitrarily small, then $\delta_3$ can be made arbitrarily close to $\delta_1$, in which case $\delta_4$ is forced to be close to $1$. Pick $\epsilon'>\delta_4$. These choices combined with Remark~\ref{rem_16}, Theorem~\ref{thm_discrete} and \eqref{e_41}, show that
$
\delta_1+\delta_2\ge 1
$.
\end{proof}
Another application of Lemma~\ref{lem_connect} is the following:
\begin{thm}[Weak converse for smooth GBLL]
\label{thm_weak}
\begin{align}
{\rm D}_{0^+}(Q_X,Q_{Y_j},c^m)
\ge\rmds(Q_X,c^m)
\label{e52}
\end{align}
\end{thm}
\begin{proof}
For simplicity, we prove for the case of $m=1$. For any $d>{\rm D}_{0^+}(Q_X,Q_Y,c)$ (achievable rate for smooth GBLL)
and any $(R,R_1)$ achievable for one-communicator CR generation,
we show that
\begin{align}
\frac{d}{c(1-\epsilon)}+R_1>R\left(1-\frac{1}{c(1-\epsilon)}\right)
\label{e_54}
\end{align}
for any $\epsilon\in(0,1)$, which will establish \eqref{e52} because of the achievable region formula \eqref{e36}.

We can choose $\delta,\delta_1,\delta_2,\delta_3,\delta_4$ such that
$\delta_2< 1-\delta-\delta_3$
and $\delta_4<\epsilon$.
For large $n$, \eqref{e_38} and \eqref{e_39} can be satisfied, and by Remark~\ref{rem15}, for $\epsilon'=\epsilon$, we can find $\mu_X$ satisfying \eqref{e_neighbor} and \eqref{e39} with $Q_X\leftarrow Q_X^{\otimes n}$, $Q_{Y|X}\leftarrow Q_{Y|X}^{\otimes n}$ and $d\leftarrow nd$.
Thus \eqref{e_54} holds because the last term in \eqref{e_41} must vanish as $n\to\infty$.
\end{proof}

\section{Converse for the Omniscient Helper Problem}
Note that Theorem~\ref{thm_weak} only establishes a weak converse for smooth GBLL and Theorem~\ref{thm_onecommunicator} is only for finite alphabets and deterministic decoders, because of the use of the blowing-up lemma. In this section we improve these results in a special case where $X=(Y_1,\dots,Y_m)$, that is, in the special case of smooth BLL and \emph{omniscient helper} CR generation.

To see why the problem becomes simpler in this special case, note that the set $\{x\colon Q_{Y|X=x}(\mathcal{A})\ge 1-\epsilon'\}$ in \eqref{e39} can be regarded as the ``preimage''
of the set $\mathcal{A}$ under the random transformation. In the case of deterministic $Q_{Y_j|X}$, there is no difference regarding the choice of $\epsilon'\in(0,1)$. However, in general a large $\epsilon'$ may imply a large $\epsilon$ on the right side of \eqref{e39}. Nevertheless, under the conditions for the blowing-up lemma, $\epsilon'$ and $\epsilon$ can be chosen independently (Remark~\ref{rem_16}).

In our prior work \cite{liu2015key}, a single-shot bound was derived via hypercontractivity which shows the strong converse property of the secret key (or CR) per unit cost.
From the current perspective, no smoothing is needed for that particular $c^m$ (which can be viewed as the orientation of the supporting hyperplane) for the reason explained in Section~\ref{sec_hypercontractivity}.
Straightforward extensions of the analysis from hypercontractivity to BLL inequality yields only a loose outer bound for the rate region when $\rmd(Q_X,Q_{Y_j},c^m)>\rmds(Q_X,c^m)$.
However, following the philosophy in the present paper, we may choose $\mu$ which is $E_1$-close to $Q_X$ and expect that $\rmd(\mu,Q_{Y_j},c^m)\approx\rmds(Q_X,c^m)$.
Thus by a slight change of the analysis in \cite{liu2015key}, we can show the following.
\begin{thm}[single-shot converse for omniscient helper CR generation]
\label{thm_oneshot}
If $d\ge \rmd(\mu, Q_{Y_j},c^m)$ for some $\mu$ satisfying
$
E_1(Q_{Y^m}\|\mu)\le\delta
$,
then
{\small
\begin{align}
\frac{1}{2}|Q_{K^m}-T_{K^m}|
\ge
1-\frac{1}{|\mathcal{K}|}
-\frac{\prod_{l=1}^m |\mathcal{W}_l|^{\frac{c_l}{\sum c_i}}}
{|\mathcal{K}|^{1-\frac{1}{\sum c_i}}}
\exp\left(\frac{d}{\sum c_i}\right)
-\delta.
\label{e38}
\end{align}
}
where $T_{K^m}(k^m):=\frac{1}{|\mathcal{K}|}1\{k_1=\dots=k_m\}$.
\end{thm}
Note that Theorem~\ref{thm_oneshot} applies for stochastic encoders and decoders, and in its proof, the function $f_j(\cdot)$ in \eqref{e_func} will take the role of $\max_w Q_{K_j|W_jY_j}(k|w,\cdot)$. However, the intuition is best explained in the case of deterministic decoders: let $\mathcal{A}_{kw_j}^j$ be the decoding set for $K_j=k$ upon receiving $w_j$ by ${\sf T}_j$. Then
\begin{align}
\mu(K_1=\dots=K_m=k)
&\le
\mu\left(\cap_j\cup_{w_j}\mathcal{A}_{kw_j}^j\right)
\\
&\le
\exp(d)\prod_j Q_{Y_j}^{c_j}\left(\cup_{w_j}\mathcal{A}_{kw_j}^j\right)
\label{e_35}
\end{align}
where the crucial step \eqref{e_35}, which may be viewed as a change-of-measure from a joint distribution to uncorrelated distributions (with powers), follows by choosing indicator functions in the BLL inequality. After some manipulations, one can bound the total variation between $\mu_{K^m}$ (consequently $Q_{K^m}$) and $T_{K^m}$.

\begin{cor}[Strong converse for omniscient helper CR generation]
Suppose $(R,R_1,\dots,R_m)$ fails \eqref{e36} for some $c^m$, and there exist a coding scheme at rates $(R,R_1,\dots,R_m)$
\begin{align}
\frac{1}{2}|Q_{K_1\dots K_m}-T_{K_1\dots K_m}|\le \delta
\label{e_62}
\end{align}
for sufficiently large $n$.
Then $\delta\ge1$ if $Q_{Y^m}$, $(Q_{Y_j|Y^m})$ and $c^m$ satisfy the smooth property (as in the case of discrete/Gaussian $Q_{Y^m}$).
\end{cor}

In the Gaussian case, refining the analysis in Theorem~\ref{thm_gauss}, we can derive a second order achievability bound for smooth BLL, which, in view of Theorem~\ref{thm_oneshot},
implies a second order converse bound for CR generation: for any sequence of CR generation schemes with non-vanishing error probability, we have
{\small
\begin{align}
\liminf_{n\to\infty}\sqrt{n}\left[\left(\sum c_j-1\right)R_n-\sum c_j R_{jn}-\rmds(Q_{Y^m},c^m)\right]
\le D\nonumber
\end{align}
}
for some constant $D$ (explicit formula given in \cite{lccv_smooth2016}), where $R_n$, $R_{1n}$, \dots, $R_{mn}$ are rates at blocklength $n$.

\begin{rem}
We used slightly different performance measures for the one-communicator problem and the omniscient helper problem.
If $\delta_1$ and $\delta_2$ satisfy \eqref{e_38}-\eqref{e_39}
then $\delta\leftarrow\delta_1+\delta_2$ satisfies \eqref{e_62},
so a strong converse measured by \eqref{e_62} implies a strong converse measured by \eqref{e_38}-\eqref{e_39}.
On the other hand,
if $\delta$ satisfies \eqref{e_62} then
$\delta_1\leftarrow\delta$ and $\delta_2\leftarrow\delta$ satisfy \eqref{e_38}-\eqref{e_39}.
Thus the strong converse in the sense of \eqref{e_38}-\eqref{e_39} only implies a ``$\frac{1}{2}$-converse'' in the sense of \eqref{e_62}.
\end{rem}
Unlike the more general one-communicator case, the rate region for omniscient helper \emph{key} generation can be obtained as the intersection of the region for omniscient helper CR generation and $\{R\le \min_j H(Y_j)\}$ \cite{liu2015key}.
(Though, the misleading similarities between the rate regions for the omniscient helper CR and key generation is only a coincidence from optimizing of the rate regions.)
As a consequence, the strong converse for the omniscient helper key generation is also proved, since the key generation counterpart obviously places more constraints, and the strong converse property of the outer-bound $\{R\le \min_j H(Y_j)\}$ is comparatively trivial.

As alluded before, the achievability for the omniscient helper CR generation implies the strong converse for smooth BLL:
\begin{cor}
\label{cor22}
For any $Q_{Y^m}$, $c^m$, and $\delta\in(0,1)$,
\begin{align}
{\rm D}_{\delta}(Q_{Y^m},Q_{Y_j},c^m)
\ge\rmds(Q_{Y^m},c^m).
\end{align}
\end{cor}
Theorem~\ref{thm_oneshot}
essentially establishes a single-shot connection between the smooth BLL and omniscient helper CR generation.
Thus the proof of Corollary~\ref{cor22} follows easily by a similar reasoning as the proof of Theorem~\ref{thm_weak}.
In fact, for a general sequence (not necessarily stationary memoryless) of sources, if the $\delta$-smooth BLL rate is strictly smaller than the supremum of $(\sum_j c_j-1)R-\sum_j c_jR_j$ over achievable rates, then the second and third terms on the right side of \eqref{e38} can be made to vanish exponentially in the blocklength. Thus $(1-\delta)$-achievability of CR generation implies $\delta$-converse for smooth BLL.

\bibliographystyle{ieeetr}
\bibliography{ref_om}

\newpage
\appendices
\section{Proof of Lemma~\ref{lem_singleletter}}
Let $I\in\{1,\dots,n\}$ be an equiprobable random variable independent of all other random variables already defined. Observe that \eqref{e_f} equals
\begin{align}
&\quad\sum_j c_jD(P_{Y_{jI}|UIY_j^{I-1}}\|\nu_j|P_{UIY_j^{I-1}})-
D(P_{X_I}\|\mu|P_{UIX^{I-1}})
\nonumber\\
&\le
\sum_j c_jD(P_{Y_{jI}|UIX^{I-1}}\|\nu_j|P_{UIX^{I-1}})-
D(P_{X_I}\|\mu|P_{UIX^{I-1}})
\label{e_singleletterize}
\end{align}
where \eqref{e_singleletterize} uses the Markov chain condition
\begin{align}
\hat{Y}_{jI}-UI\hat{X}^{I-1}-\hat{Y}_j^{I-1}.
\end{align}
Also,
$
\mathbb{E}\left[\frac{1}{n}\sum_{i=1}^n \tau_{\alpha}(\hat{X}_i)\right]\le\epsilon
$
 implies that
\begin{align}
\mathbb{E}[\tau_{\alpha}(\hat{X}_I)]\le\epsilon.
\end{align}
Therefore, with the identification
\begin{align}
P_{U,X}\leftarrow P_{UIX^{I-1},X_I}
\end{align}
we see $g(n)\le g(1)$.

\section{Proof of Lemma~\ref{lem_key}}
Each $P_{X^n}$ such that $P_{X^n}\ll \mu_n$ satisfies
\begin{align}
\mathbb{E}\left[\frac{1}{n}\sum_i \tau_{\alpha}(\hat{X}_i)\right]\le \epsilon
\end{align}
since the random variable is bounded above by $\epsilon$, $P_{X^n}$-almost surely. Then the result follows from Lemma~\ref{lem_singleletter} and the fact that $\mu_n$ and $\mu^{\otimes n}$ agree on the support of $P_{X^n}$.

\section{Proof of Lemma~\ref{lem_weakcont}}
Let $(\mathcal{B}_{\alpha})$ be any finite partition of $\mathcal{X}$ compatible with $\mathscr{F}$.
For $\alpha$ such that $Q_X(\mathcal{B}_{\alpha})>0$, define
\begin{align}
\tau_{\alpha}(x):=\frac{1\{x\in\mathcal{B}_{\alpha}\}}
{Q_X(\mathcal{B}_{\alpha})}-1,
\end{align}
and for $\alpha$ such that $Q_X(\mathcal{B}_{\alpha})=0$, put $\tau_{\alpha}=0$ if $x\notin \mathcal{B}_{\alpha}$ and $\tau_{\alpha}=\infty$ otherwise.
By the law of large numbers, the set $\mathcal{S}_{\epsilon}^n$ as defined in \eqref{e_s} satisfies
\begin{align}
\lim_{n\to\infty} Q_X^{\otimes n}(\mathcal{S}_{\epsilon}^n)=1.
\end{align}
Now we can invoke Lemma~\ref{lem_key}.
Let $\mu_n$ be the restriction of $\mu^{\otimes n}$ on $\mathcal{S}_{\epsilon}^n$,
and note that $D(P_{X^n}\|\mu^{\otimes n})=D(P_{X^n}\|\mu_n)$
By the arbitrariness of $(\mathcal{B}_{\alpha})$ and $\epsilon>0$, we see the left side of \eqref{e_smooth} is upper-bounded by
\begin{align}
\inf_{\mathcal{G},\epsilon>0}\sup_{P_X\colon P_{X|\mathcal{G}}\le(1+\epsilon) Q_{X|\mathcal{G}}}\phi(P_X)
\label{e_30}
\end{align}
where $\mathcal{G}$ is a finitely generated $\sigma$-algebra (the $\sigma$-algebra generated by $(\mathcal{B}_{\alpha})$), and $P_{X|\mathcal{G}}$ and $Q_{X|\mathcal{G}}$ are conditional distributions.
Now choose any decreasing and vanishing sequence $(\epsilon_k)$ and a nested sequence $(\mathcal{G}_k)$ which contains a countable basis of $(\mathcal{X},\mathscr{F})$. Then pick a sequence $(P_X^k)$ such that
\begin{align}
P_{X|\mathcal{G}_k}^k\le (1+\epsilon_k)Q_{X|\mathcal{G}_k}^k
\label{e_32}
\end{align}
and
\begin{align}
\lim_{k\to\infty}\phi(P_{X}^k)
=
\lim_{k\to\infty}\sup_{P_{X}\colon P_{X|\mathcal{G}_k}\le (1+\epsilon_k)Q_{X|\mathcal{G}_k}}\phi(P_{X})
\label{e_pk}
\end{align}
where the limit on the right exists by monotone convergence. By \eqref{e_32},
\begin{align}
\limsup_{k\to\infty} P_{X}^k(\mathcal{C})\le P_{X}(\mathcal{C})
\label{e30}
\end{align}
if $\mathcal{C}\in\mathcal{G}_l$ for some $l$. Since any closed subset can be constructed as the intersection of a nested sequence of such $\mathcal{C}$, it follows from the min-max inequality and the $\sigma$-continuity of probability measure that \eqref{e30} actually holds for any closed $\mathcal{C}$, establishing that $P_X^k$ converges weakly to $Q_X$. Thus the weak upper semicontinuity of $\phi(\cdot)$ and \eqref{e_pk} imply that \eqref{e_30} is bounded above by $\phi(Q_X)$, as desired.

\section{Proof of Lemma~\ref{lem_connect}}
In the $m=1$ case write $\hat{K}:=K_1$.
Define the joint measure
\begin{align}
\mu_{XYWK\hat{K}}
:=\mu_X Q_{Y|X}Q_{W|X}Q_{K|X}Q_{\hat{K}|YW}
\end{align}
which we shall sometimes abbreviate as $\mu$.
Since $E_1(Q\|\mu)=E_1(Q_X\|\mu_X)\le \delta$,
\eqref{e_38} implies
\begin{align}
\mu(K\neq \hat{K})\le \delta_1+\delta.
\end{align}
Put
\begin{align}
\mathcal{J}:=\{k\colon \mu_{\hat{K}|K}(k|k)\ge 1-\delta_4\}.
\end{align}
The Markov inequality implies that
$\mu_K(\mathcal{J})\ge 1-\delta_3$. Now for each $k\in\mathcal{J}$, we have
\begin{align}
&\quad(1-\delta_4)\mu_k(k)
\nonumber\\
&\le
\mu_{K\hat{K}}(k,k)
\\
&\le
\int_{\mathcal{F}_k}Q_{Y|X=x}\left(\bigcup_w\mathcal{A}_{wk}\right)
{\rm d}\mu_X(x)
\\
&\le
(1-\epsilon')\mu_K(k)+\mu\left(x\colon Q_{Y|X=x}\left(\bigcup_{kw}\mathcal{A}_{kw}\ge 1-\epsilon'\right)\right)
\\
&\le (1-\epsilon')\mu_K(k)
+2^c\exp(d) Q_Y^{c(1-\epsilon)}\left(\bigcup_w\mathcal{A}_{kw}\right),
\end{align}
where $\mathcal{F}_k\subseteq \mathcal{X}$ is the decoding set for $K$, and $\mathcal{A}_{kw}$ is the decoding set for $\hat{K}$ upon receiving $w$.
Rearranging,
\begin{align}
(\epsilon'-\delta_4)^{\frac{1}{c(1-\epsilon)}}\mu_k^{\frac{1}{c(1-\epsilon)}}(k) \le 2^{\frac{1}{1-\epsilon}}
\exp\left(\frac{d}{c(1-\epsilon)}\right)Q_Y\left(\bigcup_w \mathcal{A}_{kw}\right).
\label{e_50}
\end{align}
Now let $\tilde{\mu}$ be the restriction of $\mu_K$ on $\mathcal{J}$. Then summing both sides of \eqref{e_50} over $k\in\mathcal{J}$,
applying the union bound, and noting that $\{\mathcal{A}_{kw}\}_k$ is a partition of $\mathcal{Y}$ for each $w$, we obtain
\begin{align}
D_{\frac{1}{c(1-\epsilon)}}(\tilde{\mu}\|T)
&\le
\log |\mathcal{K}|-\frac{1}{1-\frac{1}{c(1-\epsilon)}}
\log\frac{2^{\frac{1}{1-\epsilon}}|\mathcal{W}|}{(\epsilon'
-\delta_4)^{\frac{1}{c(1-\epsilon)}}}
\nonumber\\
&\quad-\frac{d}{c(1-\epsilon)-1}.
\end{align}
The proof is completed invoking Proposition~\ref{prop_TVRenyi} below and noting that
\begin{align}
E_1(Q_K\|\tilde{\mu})
\le
E_1(Q_K\|\mu)
+E_1(\mu\|\tilde{\mu})\le \delta+\delta_3.
\end{align}
\begin{prop}\label{prop_TVRenyi}
Suppose $T$ is equiprobable on $\{1,\dots,M\}$ and $\mu$ is a nonnegative measure on the same alphabet. For any $\alpha\in(0,1)$,
\begin{align}
E_1(T\|\mu)
\ge
1-\frac{1}{M}-\exp(-(1-\alpha)D_{\alpha}(T\|\mu)).
\end{align}
\end{prop}
The special case of Proposition~\ref{prop_TVRenyi} when $\mu$ is a probability measure was used in the proof of \cite[Theorem~10]{liu2015key_arxiv} (see equation (59) therein) to relate R\'{e}nyi divergence and total variation distance. The extension to unnormalized $\mu$ can be easily proved in a similar way.

\section{Bound on the Second Order Rate for Gaussian Omniscient Helper CR Generation}
Let
\begin{align}
\mb{W}:=\frac{\mb{A}+\mb{A}^{\top}}{\sqrt{2}}
\end{align}
be the standard Wigner matrix, where $\mb{A}$ is a square matrix with i.i.d.~$\mathcal{N}(0,1)$ entries.

Denote by ${\rm Q}(\cdot)$ the tail probability of the standard Gaussian distribution and $\lambda_{\max}(\cdot)$ the largest eigenvalue of a matrix.

\begin{thm}[Bound on the second order rate for Gaussian omniscient helper CR generation]
\label{thm_2order}
Assume that $Q_{Y^m}$ is Gaussian with a non-degenerate covariance matrix, and there is a sequence of CR generation schemes such that
\begin{align}
&\quad\liminf_{n\to\infty}\sqrt{n}\left[\left(\sum c_j-1\right)R_n-\sum c_j R_{jn}-\rmds(Q_{Y^m},c^m)\right]
\nonumber
\\
&> \frac{\log e}{2}\left(m-\sum c_j\right)D_1+D_2
\label{e58}
\end{align}
for some $D_1,D_2\in(0,1)$, where $R_n$, $R_{1n}$, \dots, $R_{mn}$ are the rates at blocklength $n$. Then
\begin{align}
\liminf_{n\to\infty} \frac{1}{2}|Q_{{K^m}_n}-T_{{K^m}_n}|
\ge
\mathbb{P}[\lambda_{\max}(\mb{W})\le D_1]
-{\rm Q}\left(\tfrac{D_2}{\sqrt{V}}\right),
\label{e_error}
\end{align}
where
\begin{align}
V:={\rm Var}\left(
\sum_{j}c_{j}\imath_{Q_{Y_j}\|\nu_j}(Y_j) -\imath_{Q_{Y^m}\|\mu}(Y^m)\right).
\end{align}
\end{thm}
\begin{proof}
First, observe that we will only need to consider the case of $\sum_jc_j\le m$, since otherwise $\rmds(Q_{Y^m},c^m)=\infty$ and Theorem~\ref{thm_2order} is vacuous. Indeed, suppose without loss of generality that $Y^m\sim\mathcal{N}(\mb{0},\bsigma)$. For $\alpha\in(0,\infty)$ small enough, we can find $U$ jointly Gaussian with $Y^m$ such that $Y^m|U=\mb{0}\sim \mathcal{N}(\mb{0},\alpha\mb{I})$. Then we see
\begin{align}
\rmds(Q_{Y^m},c^m)
&\ge
\lim_{\alpha\downarrow0}
\sum_{j=1}^m\frac{c_j}{2}\log \frac{\sigma_{jj}}{\alpha}
-\frac{1}{2}\log\frac{|\bsigma|}{|\alpha\mb{I}|}
\\
&=\sum_{j=1}^m\frac{c_j}{2}\log\sigma_{jj}
-\frac{1}{2}\log|\bsigma|
\nonumber\\
&\quad
+\lim_{\alpha\downarrow0}\frac{\sum_j c_j-m}{2}\log\frac{1}{\alpha}
\\
&=\infty
\end{align}
provided that $\sum_jc_j>m$ holds.

The proof is essentially based on a refinement of the achievability of smooth BLL:
in the proof of Theorem~\ref{thm_gauss}, take $\epsilon_i\leftarrow \frac{D_i}{\sqrt{n}}$, $i=1,2$ and $X=Y^m$. Then,
\begin{align}
\lim_{n\to\infty}\mathbb{P}[\mathcal{S}^n_{\epsilon_1}]
&=\lim_{n\to\infty}\mathbb{P}\left[\frac{1}{n}\sum_i{\bf z}
{\bf z}^{\top}\preceq (1+\epsilon_1)\mb{I}\right]
\\
&=\lim_{n\to\infty}\mathbb{P}\left[\frac{\sum_i{\bf z}
{\bf z}^{\top}-\mb{I}}{\sqrt{n}}\preceq D_1 \mb{I}\right]
\\
&=\mathbb{P}[\mb{W}\preceq D_1 \mb{I}],
\label{e_mvclt}
\end{align}
where $\mb{z}:=\bsigma^{-\frac{1}{2}}\mb{x}\sim \mathcal{N}(\mb{0},\mb{I})$ and we applied multivariate CLT in \eqref{e_mvclt}. On the other hand, by CLT we have
\begin{align}
\lim_{n\to\infty}
\mathbb{P}[\mathcal{T}_{\epsilon_2}^n]
&=1-{\rm Q}\left(\frac{D_2}{\sqrt{V}}\right).
\end{align}
Also, a simple scaling argument shows that
\begin{align}
F((1+\epsilon_1)\bsigma)
&=F(\bsigma)+\frac{\log(1+\epsilon_1)}{2}\left(m-\sum c_j\right)
\\
&\le F(\bsigma)+\frac{\log e }{2\sqrt{n}}\left(m-\sum c_j\right)D_1.
\end{align}
Thus following the steps in the proof of Theorem~\ref{thm_gauss}, we can find $(\mu_n)_{n\ge 1}$ such that
\begin{align}
E_1(Q_X^{\otimes n}\|\mu_n)
\le\, &
\mathbb{P}[\lambda_{\max}(\mb{W})\le D_1]
-{\rm Q}\left(\tfrac{D_2}{\sqrt{V}}\right)+o(1)
\\
\frac{1}{n}\rmd(\mu_n,Q_{Y_j}^{\otimes n},c^m)
&\le \rmds(Q_X,c^m)
+\frac{\log e }{2\sqrt{n}}\left(m-\sum c_j\right)D_1
\nonumber\\
&\quad +\frac{D_2}{\sqrt{n}}.
\label{e71}
\end{align}
Now, invoke Theorem~\ref{thm_oneshot} with
\begin{align}
\mu&\leftarrow \mu_n;
\\
\delta&\leftarrow \delta_n:=E_1(Q_X^{\otimes n}\|\mu_n);
\\
d&\leftarrow nd_n,
\end{align}
where
$d_n$ is defined as the right side of \eqref{e71}.
Then
\begin{align}
\frac{1}{2}|Q_{{K^m}_n}-T_{{K^m}_n}|
&\ge
1-\frac{1}{|\mathcal{K}|}
-\exp\left(-\tfrac{\tau}{\sum_j c_j}\sqrt{n}\right)
-\delta_n
\\
&=1-\delta_n+o(1)
\end{align}
where $\tau>0$ is defined as the difference between the left and right sides of \eqref{e58}. Thus \eqref{e_error} is established.
\end{proof}

\end{document}